\documentclass[a4paper]{llncs}

\usepackage{amsfonts}
\usepackage{amsmath}
\usepackage[all]{xy}
\usepackage{algorithm}
\usepackage[noend]{algpseudocode}

\usepackage{tikz}
\usetikzlibrary{patterns}
\usepackage{subcaption} 
\usepackage{xspace}

\newcommand{\NPclassbase}{{\sf\bf NP}}

\newcommand{\hard}{\text{-hard}}

\newcommand{\NPhard}{\NPclassbase\hard\xspace}

\DeclareMathOperator{\RUN}{run}

\DeclareMathOperator{\AV}{Av}

\DeclareMathOperator{\Avd}{Av}
\newcommand\Av[2]{\Avd_{{#1}}({#2})}

\newcommand{\ptext}{\pi}

\newcommand{\bijection}{B}

\newcommand{\ppattern}{\sigma}

\newcommand\BV[2]{\genfrac{}{}{0pt}{}{#1}{#2}}

\newcounter{num}
\DeclareMathOperator{\LMEi}{LMEi}

\DeclareMathOperator{\factor}{factor}

\DeclareMathOperator{\PMa}{PM}
\newcommand{\PM}[6]{\PMa_{{#1}}^{{#2},{#3},{#4}}({#5},{#6})}

\DeclareMathOperator{\lb}{lb}
\DeclareMathOperator{\ub}{ub}

\DeclareMathOperator{\LMa}{LM}
\newcommand{\LM}[4]{\LMa_{{#1}}^{{#2}}(#3,#4)}

\DeclareMathOperator{\AFa}{AF}
\newcommand{\AF}[4]{\AFa_{{#1}}^{{#2}}(#3,#4)}

\DeclareMathOperator{\DFa}{DF}
\newcommand{\DF}[4]{\DFa_{{#1}}^{{#2}}(#3,#4)}

\DeclareMathOperator{\LCSa}{LCS}
\newcommand{\LCS}[8]{\LCSa_{{#1},{#2},{#3}}^{{#4},{#5},{#6}}({#7},{#8})}

\DeclareMathOperator{\matcha}{M}
\newcommand{\match}[8]{\matcha_{{#1},{#2},{#3}}^{{#4},{#5},{#6}}({#7},{#8})}

\DeclareMathOperator{\SETa}{S}
\newcommand{\SET}[4]{\SETa_{{#1}}^{{#2}}({#3},{#4})}

\begin{document}


\title{Pattern matching in\\ $(213,231)$-avoiding permutations}
\author{Both Emerite NEOU, Romeo RIZZI}
\date{}

\author{%
	Both Emerite Neou\thanks{On a Co-tutelle  Agreement with the Department of Mathematics of the University of Trento}\inst{1}
  Romeo Rizzi\inst{2} \and
  St\'ephane Vialette\inst{1}
}
\institute{%
	Universit\'e Paris-Est, LIGM (UMR 8049), CNRS, UPEM, ESIEE Paris, ENPC,
	F-77454, Marne-la-Vallée, France\\
  \email{\{neou,vialette\}@univ-mlv.fr}
  \and
  Department of Computer Science,
  Università degli Studi di Verona, Italy \\
  \email{romeo.rizzi@univr.it}
}

\date{\today}

\maketitle

\begin{abstract}
	Given permutations $\sigma \in S_k$ and $\pi \in S_n$ with $k<n$, the
  \emph{pattern matching} problem is to decide whether $\pi$ matches
  $\sigma$ as an order-isomorphic subsequence.
	We give a linear-time algorithm in case both $\pi$ and $\sigma$ avoid
	the two size-$3$ permutations $213$ and $231$.
	For the special case where only $\sigma$ avoids $213$ and $231$, we present a
	$O(max(kn^2,n^2\log(\log(n)))$ time algorithm. 
	We extend our research to bivincular patterns that avoid $213$ and $231$ and present a $O(kn^4)$ time algorithm.
	Finally we look at the related problem of the longest subsequence which avoids $213$ and $231$.  
\end{abstract}


\section{Introduction}
\label{section:Introduction}

	A permutation $\pi$ is said to match another permutation $\sigma$,
	in symbols $\sigma \preceq \pi$,
	if there exists a subsequence of elements of $\pi$ that has the same relative
	order as $\sigma$.
	Otherwise, $\pi$ is said to \emph{avoid} the permutation $\sigma$.
	For example a permutation matches the pattern $123$ (resp. $321$) if it has
	an increasing (resp. decreasing) subsequence of length $3$.
	As another example,
	$6152347$ matches $213$ but not $231$.
	During the last decade, the study of the pattern matching on permutations has
	become a very active area of research \cite{Kitaev:book:2011} and
	a whole annual conference (\textsc{Permutation Patterns}) is now devoted
	to this topic.

	We consider here the so-called \emph{pattern matching} problem
	(also sometimes referred to as the \emph{pattern involvement problem}):
	Given two permutations $\sigma$ and $\pi$, this problem is to decide whether
	$\sigma \preceq \pi$ (the problem is ascribed to Wilf in \cite{Bose:Buss:Lubiw:1998}).
	The permutation matching problem is known to be \NPhard~\cite{Bose:Buss:Lubiw:1998}.
	It is, however, polynomial-time solvable by brute-force enumeration
	if $\sigma$ has bounded size.
	Improvements to this algorithm were presented in
	\cite{Albert:Aldred:Atkinson:Holton:ISAAC:2001} and
	\cite{Ahal:Rabinovich:2008},
	the latter describing a $O(|\pi|^{0.47|\sigma|+o(|\sigma|)})$ time algorithm.
	Bruner and Lackner \cite{DBLP:journals/corr/abs-1204-5224}
	gave a fixed-parameter algorithm solving the pattern matching problem with
	an exponential worst-case runtime of $O(1.79^{\RUN(\pi)})$,
	where $\RUN(\pi)$ denotes the number of alternating runs of $\pi$.
	(This is an improvement upon the $O(2^{|\pi|})$ runtime required by
	brute-force search without imposing restrictions on $\sigma$ and $\pi$.)
	A recent major step was taken by Marx and Guillemot
	\cite{Guillemot:Marx:SODA:2014}.
	They showed that
	the permutation matching problem is fixed-parameter tractable (FPT) for
	parameter $|\sigma|$.

	A few particular cases of the pattern matching problem have been attacked successfully.
	The case of increasing patterns is solvable in
	$O(|\pi| \log \log |\sigma|)$ time in the RAM model \cite{Crochemore:Porat:2010},
	improving the previous 30-year bound of $O(|\pi| \log |\sigma|)$.
	Furthermore, the patterns $132$, $213$, $231$, $312$ can all be handled in linear-time
	by stack sorting algorithms.
	Any pattern of length $4$ can be detected in $O(|\pi| \log |\pi|)$ time
	\cite{Albert:Aldred:Atkinson:Holton:ISAAC:2001}.
	Algorithmic issues for $321$-avoiding patterns matching has been investigated in
	\cite{Guillemot:Vialette:ISAAC:2009}.
	The pattern matching problem is also solvable in
	polynomial-time for separable patterns \cite{Ibarra:1997,Bose:Buss:Lubiw:1998}
	(see also \cite{Bouvel:Rossin:Vialette:CPM:2007} for LCS-like issues
	of separable permutations).
	Separable permutations are those permutations that match neither
	$2413$ nor $3142$, and they are enumerated by the Schr{\"o}der numbers
	(Notice that the separable permutations include as a special case the
	stack-sortable permutations, which avoid the pattern $231$.)

	There exists many generalisation of patterns that are worth considering
	in the context of algorithmic issues in pattern matching
	(see \cite{Kitaev:book:2011} for an up-to-date survey).
	\emph{Vincular patterns}, also called
	\emph{generalized patterns},
	resemble (classical) patterns with the additional constraint that some of the elements in
	a matching must be consecutive in postitions.
	Of particular importance in our context,
	Bruner and Lackner \cite{DBLP:journals/corr/abs-1204-5224}
	proved that deciding whether a vincular pattern
	$\sigma$ of length $k$ can be match to a longer permutation
	$\pi$ is $W[1]$-complete for
	parameter $k$;
	for an up to date survey of the $W[1]$ class and related material, see
	\cite{Downey:Fellows:2013}.
	\emph{Bivincular patterns} generalize classical patterns even further
	than vincular
	patterns by adding a constraint on values.

	We focus in this paper on pattern matching issues for
	$(213,231)$-avoiding permutations
    (\emph{i.e.}, those permutations that avoid both $213$ and $231$).
	The number of $n$-permutations that avoid both
	$213$ and $231$ is
	$t_0 = 1$ for $n = 0$ and
	$t_n =2^{n-1}$ for $n\geq 1$ \cite{Simion:Schmidt:EJC:1985}.
	On an individual basis,
	the permutations that do not match the permutation pattern $231$
	are exactly the \emph{stack-sortable permutations} and they are counted by
	the Catalan numbers \cite{Knuth:1997:ACP:260999}.
	A stack-sortable permutation is a permutation whose elements may be sorted by
	an algorithm whose internal storage is limited to a single stack data structure.
	As for $213$, it is well-known that
 	if $\pi = \pi_1\pi\,\ldots\,\pi_n$ avoids $132$, then its complement
 	$\pi' = (n+1-\pi_1)(n+1-\pi_2)\,\ldots\,(n+1-\pi_n)$ avoids $312$, and
 	the reverse of $\pi'$ avoids $213$.
	This paper is organized as follows.
	In Section~\ref{section:Definitions} the needed definitions are presented.
	Section~\ref{section:both are (213,231)-avoiding} is devoted to presenting
	an online linear-time algorithm in case both 
	permutations are $(213,231)$-avoiding,
	whereas Section~\ref{section:sigma only avoid 231 and 213} focuses on the case
	where only the pattern is $(213,231)$-avoiding.
	In Section~\ref{section:bivincular} we give a polynomial-time algorithm
	for $(213,231)$-avoiding bivincular patterns.
	In Section~\ref{section:LCS} we consider the problem of finding the longest
	$(213,231)$-avoiding pattern in permutations.


\section{Definitions}
\label{section:Definitions}

A \emph{permutation} of length $n$ is a one-to-one function from an
$n$-element set to itself.
We write permutations as words
$\pi = \pi_1\pi_2\,\ldots\,\pi_n$, whose elements are distinct
and usually consist of the integers $12\,\ldots\,n$, and we let
$\pi[i]$ stands for $\pi_i$.
For the sake of convenience, we let
$\pi[i:j]$ stand for
$\pi_i\pi_{i+1}\,\ldots\,\pi_j$,
$\pi[:j]$ stand for $\pi[1:j]$ and
$\pi[i:]$ stand for $\pi[i:n]$.
As usual, we let $S_n$ denote the set of all permutations of length $n$.
It is convenient to use a geometric representation of permutation
to ease the understanding of algorithms. 
The geometric representation corresponds to the set of point with coordinate $(i,\pi[i])$
(see figure \ref{example: pattern matching}).


A permutation $\pi$ is said to \emph{match} the permutation $\sigma$
if there exists a subsequence of (not necessarily consecutive)
element of $\pi$ that has the same relative order as $\sigma$,
and in this case $\pi$ is said to match $\sigma$, 
written $\sigma \preceq \pi$.
Otherwise, $\pi$ is said to \emph{avoid} the permutation $\sigma$.
For example, the permutation $\pi = 391867452$
matches the pattern $\sigma = 51342$,
as can be seen in the highlighted subsequence of
$\pi = 3\mathbf{9}\mathbf{1}8\mathbf{6}\mathbf{7}\mathbf{4}52$
(or
$\pi = 3\mathbf{9}\mathbf{1}8\mathbf{6}\mathbf{7}4\mathbf{5}2$
or
$\pi = 3\mathbf{9}\mathbf{1}8\mathbf{6}\mathbf{7}45\textbf{2}$
or
$\pi = 3\mathbf{9}\mathbf{1}867\textbf{4}\textbf{5}\mathbf{2}$
).
Each subsequence $91674$,
$91675$,
$91672$,
$91452$,
 in $\pi$ is called a
\emph{matching}
of $\sigma$.
Since the permutation $\pi = 391867452$  contains no increasing subsequence of
length four, $\pi$ avoids $1234$.
Geometrically, $\pi$ matches $\sigma$ if there exists
a set of point in $\pi$ that is isomorph to the set of point of $\sigma$. In other word,
if there exists a set of point in $\pi$ with the same disposition as the set of point of $\sigma$, without regard to the distance (see figure \ref{example: pattern matching}).

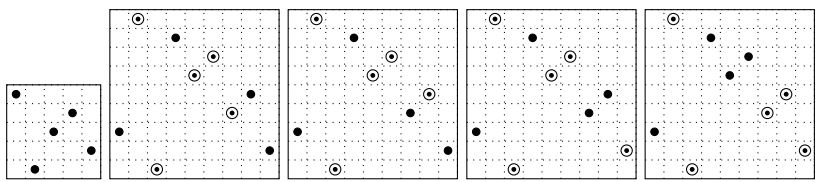
\begin{figure}
	\centering
		\begin{tikzpicture}[scale=.25]
		\draw (0,0) -- (5,0) -- (5,5) -- (0,5) -- cycle;
		\foreach \i in {1, 2, 3, 4, 5}{
			\draw[dotted] (0,\i) -- (5, \i);
			\draw[dotted] (\i,0) -- (\i, 5);
		}
		\draw[fill=black] (0.5,4.5) circle (2mm);
		\draw[fill=black] (1.5,0.5) circle (2mm);
		\draw[fill=black] (2.5,2.5) circle (2mm);
		\draw[fill=black] (3.5,3.5) circle (2mm);
		\draw[fill=black] (4.5,1.5) circle (2mm);
		\end{tikzpicture} 
		\begin{tikzpicture}[scale=.25]
		\draw (0,0) -- (9,0) -- (9,9) -- (0,9) -- cycle;
		\foreach \i in {1, 2, 3, 4, 5, 6, 7, 8, 9}{
			\draw[dotted] (0,\i) -- (9, \i);
			\draw[dotted] (\i,0) -- (\i, 9);
		}
		\draw[fill=black] (0.5,2.5) circle (2mm);
		\draw[fill=black,double, double distance=1pt] (1.5,8.5) circle (2mm);
		\draw[fill=black,double, double distance=1pt] (2.5,0.5) circle (2mm);
		\draw[fill=black] (3.5,7.5) circle (2mm);
		\draw[fill=black,double, double distance=1pt] (4.5,5.5) circle (2mm);
		\draw[fill=black,double, double distance=1pt] (5.5,6.5) circle (2mm);
		\draw[fill=black,double, double distance=1pt] (6.5,3.5) circle (2mm);
		\draw[fill=black] (7.5,4.5) circle (2mm);
		\draw[fill=black] (8.5,1.5) circle (2mm);
		\end{tikzpicture}
		\begin{tikzpicture}[scale=.25]
		\draw (0,0) -- (9,0) -- (9,9) -- (0,9) -- cycle;
		\foreach \i in {1, 2, 3, 4, 5, 6, 7, 8, 9}{
			\draw[dotted] (0,\i) -- (9, \i);
			\draw[dotted] (\i,0) -- (\i, 9);
		}
		\draw[fill=black] (0.5,2.5) circle (2mm);
		\draw[fill=black,double, double distance=1pt] (1.5,8.5) circle (2mm);
		\draw[fill=black,double, double distance=1pt] (2.5,0.5) circle (2mm);
		\draw[fill=black] (3.5,7.5) circle (2mm);
		\draw[fill=black,double, double distance=1pt] (4.5,5.5) circle (2mm);
		\draw[fill=black,double, double distance=1pt] (5.5,6.5) circle (2mm);
		\draw[fill=black] (6.5,3.5) circle (2mm);
		\draw[fill=black,double, double distance=1pt] (7.5,4.5) circle (2mm);
		\draw[fill=black] (8.5,1.5) circle (2mm);
		\end{tikzpicture}
		\begin{tikzpicture}[scale=.25]
		\draw (0,0) -- (9,0) -- (9,9) -- (0,9) -- cycle;
		\foreach \i in {1, 2, 3, 4, 5, 6, 7, 8, 9}{
			\draw[dotted] (0,\i) -- (9, \i);
			\draw[dotted] (\i,0) -- (\i, 9);
		}
		\draw[fill=black] (0.5,2.5) circle (2mm);
		\draw[fill=black,double, double distance=1pt] (1.5,8.5) circle (2mm);
		\draw[fill=black,double, double distance=1pt] (2.5,0.5) circle (2mm);
		\draw[fill=black] (3.5,7.5) circle (2mm);
		\draw[fill=black,double, double distance=1pt] (4.5,5.5) circle (2mm);
		\draw[fill=black,double, double distance=1pt] (5.5,6.5) circle (2mm);
		\draw[fill=black] (6.5,3.5) circle (2mm);
		\draw[fill=black] (7.5,4.5) circle (2mm);
		\draw[fill=black,double, double distance=1pt] (8.5,1.5) circle (2mm);
		\end{tikzpicture}		
		\begin{tikzpicture}[scale=.25]
		\draw (0,0) -- (9,0) -- (9,9) -- (0,9) -- cycle;
		\foreach \i in {1, 2, 3, 4, 5, 6, 7, 8, 9}{
			\draw[dotted] (0,\i) -- (9, \i);
			\draw[dotted] (\i,0) -- (\i, 9);
		}
		\draw[fill=black] (0.5,2.5) circle (2mm);
		\draw[fill=black,double, double distance=1pt] (1.5,8.5) circle (2mm);
		\draw[fill=black,double, double distance=1pt] (2.5,0.5) circle (2mm);
		\draw[fill=black] (3.5,7.5) circle (2mm);
		\draw[fill=black] (4.5,5.5) circle (2mm);
		\draw[fill=black] (5.5,6.5) circle (2mm);
		\draw[fill=black,double, double distance=1pt] (6.5,3.5) circle (2mm);
		\draw[fill=black,double, double distance=1pt] (7.5,4.5) circle (2mm);
		\draw[fill=black,double, double distance=1pt] (8.5,1.5) circle (2mm);
		\end{tikzpicture}

	\caption[Example pattern matching]{
		The pattern $\sigma=51342$
		and four matchings of $\sigma$ in $391867452$.} 
	\label{example: pattern matching}
\end{figure}

Suppose $P$ is a set of permutations. We let $\AV_n(P)$ denote the
set of all $n$-permutations avoiding each permutation in $P$.
For the sake of convenience
(and as it is customary~\cite{Kitaev:book:2011}), we omit $P$'s braces thus having
e.g. $\AV_n(213,231)$ instead of
$\AV_n(\{213,231\})$.
If $\pi \in \AV_n(P)$, we also say that $\pi$ is
\emph{$P$-avoiding}.

An \emph{ascent} of a permutation $\pi \in S_n$ is any element
$1 \leq i < n$ where the following value is bigger than the current one.
That is, if $\pi = \pi_1\pi_2\,\ldots\,\pi_n$, then
$\pi_i$ is an ascent if $\pi_i < \pi_{i+1}$.
For example, the permutation
$3452167$ has ascents $3$, $4$, $1$ and $6$.
Similarly, a \emph{descent} is any element
$1 \leq i < n$ with $\pi_i > \pi_{i+1}$,
so for every $1 \leq i < n$, $\pi_i$  is either an ascent or is a descent of
$\pi$.

A $\textit{left to right maxima}$ (abbreviate LRMax) of $\pi$ is a element that does not have any element bigger than it on its right (see fig. $\ref{fig:left to right maxima}$). Formally, $\pi[i]$ is a LRMax if and only if $\pi[i]$ is the biggest element of $\pi[i:]$. Similarly $\pi[i:]$ is a $\textit{left to right minima}$ (abbreviate LRMin) if and only if $\pi[i]$ is the smallest element of $\pi[i:]$.

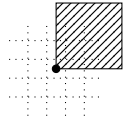
\begin{figure}
	\centering
	\begin{tikzpicture}[scale=.25]
	\foreach \i in {1, 2, 3, 4}{
		\draw[dotted] (0,\i) -- (5, \i);
		\draw[dotted] (\i,0) -- (\i, 5);
	}
	\draw[fill=black] (2.5,2.5) circle (2mm);
	
	\draw [pattern=north east lines] (2.5,2.5) rectangle (6,6);
	
	\end{tikzpicture}

	\caption[Example pattern matching]{
		The element is a LRMax if and only if the dashed area is empty.} 
	\label{fig:left to right maxima}
\end{figure}

A \emph{bivincular pattern} $\sigma$
of length $k$ is a permutation in $S_k$ written in
two-line notation
(that is the top row is $12\,\ldots\,k$ and the bottom row
is a permutation $\sigma_1\sigma_2\,\ldots\,\sigma_k$).
We have the following conditions on the top and bottom rows
of $\sigma$, as see in \cite{Kitaev:book:2011} in definition 1.4.1:
\begin{itemize}
	\item
	If the bottom line of $\sigma$ contains
	$\underline{\sigma_i\sigma_{i+1}\,\ldots\,\sigma_j}$
	then the elements corresponding to
	$\sigma_i\sigma_{i+1}\,\ldots\,\sigma_j$ in a matching of
	$\sigma$ in $\pi$ must be adjacent, whereas there is
	no adjacency condition for
	non-underlined consecutive elements.
	Moreover if the bottom row of $\sigma$ begins with
	$_\llcorner{\sigma_1}$ then any matching of $\sigma$
	in a permutation $\pi$ must begin with the leftmost
	element of $\pi$,
	and
	if the bottom row of $\sigma$ begins with
	${\sigma_k}_\lrcorner$ then any matching of $\sigma$
	in a permutation $\pi$ must end with the rightmost
	element of $\pi$.
	\item
	If the top line of $\sigma$ contains
	$\overline{i\,i+1\,\ldots\,j}$ then the elements corresponding to
	$i, i+1, \ldots, j$ in an
	matching of $\sigma$ in $\pi$ must be adjacent in values,
	whereas there is no value adjacency restriction for non-overlined
	elements.
	Moreover, if the top row of $\sigma$ begins with
	$^\ulcorner{1}$ then
	any matching of $\sigma$ is a permutation $\pi$ must contain 
	the smallest element of $\pi$, and
	if top row of $\sigma$ ends with $k^\urcorner$ then
	any matching of $\sigma$ is a permutation $\pi$ must contain
	the largest element of $\pi$.
	
\end{itemize}

For example,
let
$\sigma = \BV{1\overline{23}4\urcorner}{\llcorner 21\underline{43}  }$.
In $3217845$, $\textbf{32}17\textbf{84}5$ is a matching of $\sigma$ but
 $3\textbf{21}\textbf{7}8\textbf{4}5$ is not.
The best general reference is \cite{Kitaev:book:2011}.

Geometrically, We represent underlined and overlined elements by \textit{forbidden areas}.
A vertical area between two points indicates that the two matching of those points must be consecutive in positions, whereas a horizontal area between two points indicates that the two  matching of those points must be consecutive in value. The forbidden areas can be understand as follow : in a matching, the forbidden areas must be empty. Thus, 
$\pi$ matches a bivincular pattern $\sigma$ if there exists a set of point in $\pi$ that is isomorph to $\sigma$ and if
the forbidden areas are empty.
(see figure \ref{example:bivincular pattern matching}).

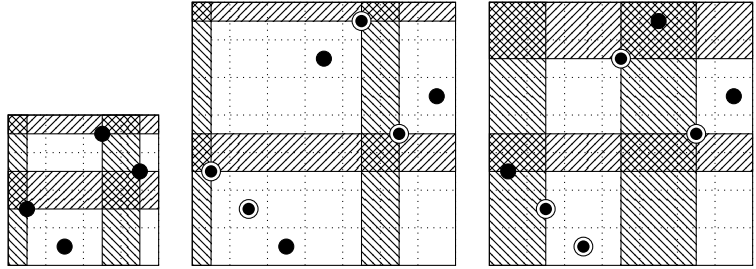
\begin{figure}[t] 
	\centering

    	\begin{tikzpicture}[scale=.5]
    	\draw [pattern=north east lines] (0,1.5) rectangle (4,2.5);
    	\draw [pattern=north east lines] (0,3.5) rectangle (4,4);
    	
    	\draw [pattern=north west lines] (0,0) rectangle (0.5,4);
    	\draw [pattern=north west lines] (2.5,0) rectangle (3.5,4);

    	\draw (0,0) -- (4,0) -- (4,4) -- (0,4) -- cycle;
    	\foreach \i in {1, 2, 3, 4}{
    		\draw[dotted] (0,\i) -- (4, \i);
    		\draw[dotted] (\i,0) -- (\i, 4);
    	}
    	\draw[fill=black] (0.5,1.5) circle (2mm);
    	\draw[fill=black] (1.5,0.5) circle (2mm);
    	\draw[fill=black] (2.5,3.5) circle (2mm);
    	\draw[fill=black] (3.5,2.5) circle (2mm);
    	\end{tikzpicture}
    	\quad
    	\begin{tikzpicture}[scale=.5]
    	\draw [pattern=north east lines] (0,6.5) rectangle (7,7);
    	\draw [pattern=north east lines] (0,2.5) rectangle (7,3.5);
    	
    	\draw [pattern=north west lines] (0,0) rectangle (0.5,7);
    	\draw [pattern=north west lines] (4.5,0) rectangle (5.5,7);
    	
    	
    	\draw (0,0) -- (7,0) -- (7,7) -- (0,7) -- cycle;
    	\foreach \i in {1, 2, 3, 4, 5, 6, 7}{
    		\draw[dotted] (0,\i) -- (7, \i);
    		\draw[dotted] (\i,0) -- (\i, 7);
    	}
    	
		\draw[fill=black,double, double distance=1pt] (0.5,2.5) circle (2mm);
		\draw[fill=black,double, double distance=1pt] (1.5,1.5) circle (2mm);
		\draw[fill=black] (2.5,0.5) circle (2mm);
		\draw[fill=black] (3.5,5.5) circle (2mm);
		\draw[fill=black,double, double distance=1pt] (4.5,6.5) circle (2mm);
		\draw[fill=black,double, double distance=1pt] (5.5,3.5) circle (2mm);
		\draw[fill=black] (6.5,4.5) circle (2mm);	
    	
    	\end{tikzpicture}  
    	\quad
    	\begin{tikzpicture}[scale=.5]
    	\draw [pattern=north east lines] (0,5.5) rectangle (7,7);
    	\draw [pattern=north east lines] (0,2.5) rectangle (7,3.5);
    	
    	\draw [pattern=north west lines] (0,0) rectangle (1.5,7);
    	\draw [pattern=north west lines] (3.5,0) rectangle (5.5,7);
    	
    	
    	\draw (0,0) -- (7,0) -- (7,7) -- (0,7) -- cycle;
    	\foreach \i in {1, 2, 3, 4, 5, 6, 7}{
    		\draw[dotted] (0,\i) -- (7, \i);
    		\draw[dotted] (\i,0) -- (\i, 7);
    	}
    	
    	\draw[fill=black] (0.5,2.5) circle (2mm);
    	\draw[fill=black,double, double distance=1pt] (1.5,1.5) circle (2mm);
    	\draw[fill=black,double, double distance=1pt] (2.5,0.5) circle (2mm);
    	\draw[fill=black,double, double distance=1pt] (3.5,5.5) circle (2mm);
    	\draw[fill=black] (4.5,6.5) circle (2mm);
    	\draw[fill=black,double, double distance=1pt] (5.5,3.5) circle (2mm);
    	\draw[fill=black] (6.5,4.5) circle (2mm);	
    	
    	\end{tikzpicture}

    	\caption[Example pattern matching]{
    		From left to right,
    		the bivincular pattern $\sigma = \BV{1\overline{23}4\urcorner}{\llcorner 21\underline{43}  }$, A matching of $\sigma$ in $3216745$, A matching of $2143$ in $3216745$ but not a matching of $\sigma$ in $3216745$ because the point $(1,3)$ and $(5,7)$ are in the forbidden area.} 
    	\label{example:bivincular pattern matching}
\end{figure}


\section{Both $\pi$ and $\sigma$ are $(213,231)$-avoiding}
\label{section:both are (213,231)-avoiding}

This section is devoted to presenting a fast algorithm for deciding if
$\sigma \preceq \pi$
in case both $\pi$ and $\sigma$ are $(213,231)$-avoiding.
We begin with an easy but crucial structure lemma.

\begin{lemma}[Folklore]
\label{lemma:first element is 1 or n}
The first element of any $(213,231)$-avoiding permutation
must be either the minimal or the maximal element.
\end{lemma}

\begin{proof}[of Lemma~\ref{lemma:first element is 1 or n}]
Any other initial element would serve as a `$2$' in either a
$231$ or $213$ with $1$ and $n$ as the `$1$' and `$3$' respectively.
\qed
\end{proof}

\begin{corollary}
\label{corollary:minmaxelement}
$\pi \in \AV_n(213,231)$ if and only if for $1 \leq i < n$,
$\pi[i]$ is a LRMax or a LRMin.
\end{corollary}

\begin{corollary}
\label{corollary:max is ascent}
Let $\pi \in \AV_n(213,231)$ and $1 \leq i < n$. Then,
(1) 
$\pi[i]$ is an ascent element if and only if $\pi[i]$ is a LRMin
and 
(2)
$\pi[i]$ is a descent element if and only if $\pi[i]$ is a LRMax
\end{corollary}

Lemma $\ref{corollary:max is ascent}$ gives a bijection between $\AV_n(213,231)$ and the set of binary word of size $n-1$. The bijected word $w$ of $\pi$, is the word where each letter at position $i$ represents if $\pi[i]$ is an ascent or descent element (or is a LRMax or a LRMin). We call this bijection $\bijection$.

A $(213,231)$-avoiding permutation has a particular form.
If we take only the descent elements, the points draw a north-east to south-west line
and if we take only the ascent elements, the points draw a south-east to north-west line.
This shape the permutation as a $>$. For convenience when drawing a random $(213,231)$-avoiding permutation
we will sometime represent a sequence of ascent/descent element by lines (see figure \ref{fig:shape of the permutation} and \ref{fig:shape of the permutation plus factor}).

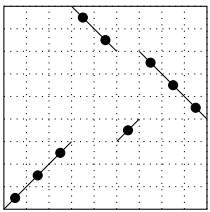
\begin{figure}[t] 
	\centering
	\begin{tikzpicture}[scale=.3]
	\draw (0,0) -- (9,0) -- (9,9) -- (0,9) -- cycle;
	\foreach \i in {1, 2, 3, 4, 5, 6, 7, 8, 9}{
		\draw[dotted] (0,\i) -- (9, \i);
		\draw[dotted] (\i,0) -- (\i, 9);
	}
	\draw[fill=black] (0.5,0.5) circle (2mm);
	\draw[fill=black] (1.5,1.5) circle (2mm);
	\draw[fill=black] (2.5,2.5) circle (2mm);
	\draw[fill=black] (3.5,8.5) circle (2mm);
	\draw[fill=black] (4.5,7.5) circle (2mm);
	\draw[fill=black] (5.5,3.5) circle (2mm);
	\draw[fill=black] (6.5,6.5) circle (2mm);
	\draw[fill=black] (7.5,5.5) circle (2mm);
	\draw[fill=black] (8.5,4.5) circle (2mm);
	
	\draw (0,0) -- (3,3);
	\draw (3,9) -- (5,7);
	\draw (5,3) -- (6,4);
	\draw (6,7) -- (9,4);

	\end{tikzpicture}

	\caption{The $(213,231)$-avoiding permutation $123984765$, every point which is on a north-east to south-west line represents a descent element and every point which is on a south-east to north-west line represents an ascent element.}
	\label{fig:shape of the permutation}
\end{figure}	

The following lemma is central to our algorithm.

\begin{lemma}
\label{lemma:MatchStripeToPermutation}
Let $\pi$ and $\sigma$ be two $(213,231)$-avoiding permutations,
Then, $\pi$ matches $\sigma$ if and only if 
there exists a subsequence $t$ of $\pi$ such $\bijection(t)=\bijection(\sigma)$.
\end{lemma}

\begin{proof}[of Lemma~\ref{lemma:MatchStripeToPermutation}]
The forward direction is obvious.
We prove the backward direction by induction on the size of the pattern
$\sigma$.
The base case is a pattern of size $2$.
Suppose that $\sigma = 12$ and thus $\bijection(\sigma) = ascent$.
Let $t = \pi_{i_1}\pi_{i_2}$, $i_1 < i_2$, be a subsequence of $\pi$
such that $\bijection(t) = ascent$, this reduces to saying that
$\pi_{i_1} < \pi_{i_2}$, 
and hence that $t$ is a matching of $\sigma = 12$ in $\pi$.
A similar argument shows that the lemma holds true for $\sigma = 21$.
Now, assume that the lemma is true for all patterns up to size $k \geq 2$.
Let $\sigma \in \Av{k+1}{231,213}$ and
let $t$,
be a subsequence of $\pi$ of length $k+1$ such that
$\bijection(t) = \bijection(\sigma)$.
As $\bijection(t)[2:] = \bijection(\sigma)[2:]$
by the inductive hypothesis, it follows that
$t[2:]$ is a matching of $\sigma[2:]$.
Moreover $\bijection(t)[1] = \bijection(\sigma)[1]$ 
thus $t[1]$ and $\sigma[1]$ are both either the minimal or the maximal
element of their respective subsequences.
Therefore, $t$ is a matching of $\sigma$ in $\pi$.
\qed
\end{proof}

We are now ready to solve the pattern matching problem in case
both $\pi$ and $\sigma$ are $(213, 231)$-avoiding.

\begin{proposition}
	\label{Proposition:both permutations are avoiding}
	Let $\pi$ and $\sigma$ be two $(213,231)$-avoiding permutations.
	One can decide whether $\pi$ matches $\sigma$ in linear time.
\end{proposition}

\begin{proof}
According to Lemma~\ref{Proposition:both permutations are avoiding} the problem reduces
to deciding whether $s_\sigma$ occurs as a subsequence in $s_\pi$.
A straightforward greedy approach solves this issue in linear-time.
\qed
\end{proof}

Thank to Corollary~\ref{corollary:max is ascent},
we can compute the bijected words in the same time that we running the greedy algorithm, this gives us a on-line algorithm.


\section{$\sigma$ only is $(213,231)$-avoiding}
\label{section:sigma only avoid 231 and 213}

This section focuses on the pattern matching problem
in case only the pattern $\sigma$ avoids $231$ and $213$.
We need to consider a specific decomposition of $\sigma$ into \textit{factor} :
we split the permutation into largest sequences of ascent element and descent element, respectively called an ascent factor and a descent factor.
This correspond to split the permutation between every pair of ascent-descent element and descent-ascent element (see figure \ref{fig:shape of the permutation plus factor}).
For the special case of $(213,231)$-avoiding, this correspond to split the permutation into largest sequence of consecutive element.
We will label the factors from right to left.

We introduce the notation $\LMEi(s)$ : Suppose that $s$ is a subsequence of $S$, $\LMEi(s)$ is the index of the left most element of $s$ in $S$. Thus
for every $\factor$, $\LMEi(\factor(j))$ stand for the index in $\sigma$
of the leftmost element of $\factor(j)$.
For example,
$\sigma = 123984765$ is split as
 $123-98-4-765$. 
Hence
$\sigma =$ $\factor(4)$ $\factor(3)$ $\factor(2)$ $\factor(1)$ with
$\factor(4) =123$, $\factor(3) = 98$, $\factor(2) = 4$ and $\factor(1) = 765$.
Furthermore,
$\LMEi(\factor(4)) = 1$, $\LMEi(\factor(3)) = 4$, $\LMEi(\factor(2)) = 6$ and $\LMEi(\factor(1)) = 7$.
We represent elements matching an ascent (resp. descent) factor 
by a rectangle which has the left most matched point of the factor as the left bottom (resp. top) corner and the right most matched point of the factor as the right top (resp. bottom) corner.
We can remove the two right most rectangles and replace it by the smallest rectangle that contains both of them and repeat this operation to represent part of a matching (see figure \ref{fig:shape of the permutation plus factor plus factor}).

\begin{remark}
A factor is either an increasing or a decreasing sequence of element. Thus while matching a factor, it is enough to find an increasing or a decreasing subsequence of same size or bigger than the factor.
\end{remark}

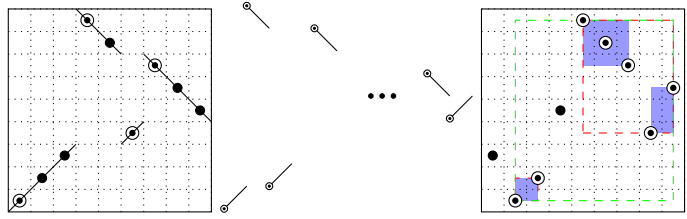
\begin{figure}[t] 
	\centering
	\begin{tikzpicture}[scale=.3]
	\draw (0,0) -- (9,0) -- (9,9) -- (0,9) -- cycle;
	\foreach \i in {1, 2, 3, 4, 5, 6, 7, 8, 9}{
		\draw[dotted] (0,\i) -- (9, \i);
		\draw[dotted] (\i,0) -- (\i, 9);
	}
	\draw[fill=black,double, double distance=1pt] (0.5,0.5) circle (2mm);
	\draw[fill=black] (1.5,1.5) circle (2mm);
	\draw[fill=black] (2.5,2.5) circle (2mm);
	\draw[fill=black,double, double distance=1pt] (3.5,8.5) circle (2mm);
	\draw[fill=black] (4.5,7.5) circle (2mm);
	\draw[fill=black,double, double distance=1pt] (5.5,3.5) circle (2mm);
	\draw[fill=black,double, double distance=1pt] (6.5,6.5) circle (2mm);
	\draw[fill=black] (7.5,5.5) circle (2mm);
	\draw[fill=black] (8.5,4.5) circle (2mm);
	
	\draw (0,0) -- (3,3);
	\draw (3,9) -- (5,7);
	\draw (5,3) -- (6,4);
	\draw (6,7) -- (9,4);
	
	\end{tikzpicture}	
	\begin{tikzpicture}[scale=.3]

	\draw (0,0) -- (1,1);
	\draw (1,9) -- (2,8);
	\draw (2,1) -- (3,2);
	\draw (4,8) -- (5,7);
	
	\draw[fill=black] (6.5,5) circle (1mm);
	\draw[fill=black] (7,5) circle (1mm);
	\draw[fill=black] (7.5,5) circle (1mm);
	
	\draw (9,6) -- (10,5);
	\draw (10,4) -- (11,5);
	
	\draw[fill=black,double, double distance=0.5pt] (0,0) circle (1mm);
	\draw[fill=black,double, double distance=0.5pt] (1,9) circle (1mm);
	\draw[fill=black,double, double distance=0.5pt] (2,1) circle (1mm);
	\draw[fill=black,double, double distance=0.5pt] (4,8) circle (1mm);
	\draw[fill=black,double, double distance=0.5pt] (9,6) circle (1mm);
	\draw[fill=black,double, double distance=0.5pt] (10,4) circle (1mm);
	
	\end{tikzpicture}	
	\begin{tikzpicture}[scale=.3]
	
	\fill[blue!40!white] (1.5,1.5) rectangle (2.5,0.5);
	\fill[blue!40!white] (4.5,8.5) rectangle (6.5,6.5);
	\fill[blue!40!white] (7.5,3.5) rectangle (8.5,5.5);
	
	\draw[red,dashed] (1.5,1.5) rectangle (2.5,0.5);
	\draw[red,dashed] (4.5,8.5) rectangle (8.5,3.5);
	
	\draw[green,dashed] (1.5,0.5) rectangle (8.5,8.5);

	\draw (0,0) -- (9,0) -- (9,9) -- (0,9) -- cycle;
	\foreach \i in {1, 2, 3, 4, 5, 6, 7, 8, 9}{
		\draw[dotted] (0,\i) -- (9, \i);
		\draw[dotted] (\i,0) -- (\i, 9);
	}
	\draw[fill=black] (0.5,2.5) circle (2mm);
	\draw[fill=black,double, double distance=1pt] (1.5,0.5) circle (2mm);
	\draw[fill=black,double, double distance=1pt] (2.5,1.5) circle (2mm);
	\draw[fill=black] (3.5,4.5) circle (2mm);
	\draw[fill=black,double, double distance=1pt] (4.5,8.5) circle (2mm);
	\draw[fill=black,double, double distance=1pt] (5.5,7.5) circle (2mm);
	\draw[fill=black,double, double distance=1pt] (6.5,6.5) circle (2mm);
	\draw[fill=black, double distance=1pt] (7.5,3.5) circle (2mm);
	\draw[fill=black, double distance=1pt] (8.5,5.5) circle (2mm);

	\end{tikzpicture}				
	
	\caption{The left figure is the  $(213,231)$-avoiding permutation $123984765$. Every line represents a factor, every circled point represents the left most element of each factor. 
	The central figure is a generalisation of $(213,231)$-avoiding permutation.	
	The right figure represents a matching of $1276534$ in $312598746$, the blue rectangles represent the matching, the red dotted rectangles represent the matching after the first replacement and the green rectangle represents the matching after the second replacement.}
	\label{fig:shape of the permutation plus factor}
\end{figure}

\begin{corollary}
\label{corollary:whereIsMax}
Given a permutation $\sigma \in \AV(213,231)$ and a suffix of its 
decomposition $\factor(i)$ $\factor(i-1)$ $\ldots$ $\factor(1)$, 
if $\factor(i)$ is an ascent (respectively descent) factor 
then the maximal (resp. minimal) element of $\factor(i)$ $\factor(i-1)$ $\ldots$ $\factor(1)$
is the left most element of $\factor(i-1)$
\end{corollary}

This is a corollary of lemma \ref{lemma:first element is 1 or n}.
This states that given a permutation in $\AV(213,231)$ 
if the permutation starts with ascent (respectively descent) elements then the maximal (resp. minimal) element of this permutation is the first descent (resp. ascent) element (see figure \ref{fig:where to find the max}).

\begin{figure}[t] 
	\centering	
	\begin{tikzpicture}[scale=.3]

	\draw (0,0) -- (1,1);
	\draw (2,10) -- (3,9);
	
	\draw[fill=black] (3.5,5) circle (1mm);
	\draw[fill=black] (4,5) circle (1mm);
	\draw[fill=black] (4.5,5) circle (1mm);
	
	\draw [blue] (6,2) -- (7,3);
	\draw (7,7) -- (8,6);
	\draw (8,3) -- (9,4);
	\draw (9,6) -- (10,5);
	\draw (10,4) -- (11,5);

	\draw[fill=black] (7,7) circle (2mm);
	\draw[fill=blue] (6,2) circle (2mm);
	
	\draw[dotted,-] (0,7) -- (13,7) node[right] {Maximal y-coordinate};
	\draw[dotted,-] (0,2) -- (13,2) node[right] {Minimal y-coordinate};
	
	\draw[|->] (6,0) -- (13,0) node[below] {Suffix starting at $\factor(i)$};

	
	\end{tikzpicture}

	\caption{
		The maximal element of the suffix starting at  $\factor(i)$ (represented by the blue line)
		is the left most element of $\factor(i-1)$ (represented by the the black dot).
		$\LM{\ppattern}{\ptext}{i}{j}$ is the minimal value of the matching of the maximal element (the black dot) in all the matchings of the suffix starting at $\LMEi(\factor(i))$ (the blue dot) in $\pi[j:]$.
	}	
	\label{fig:where to find the max}
\end{figure}
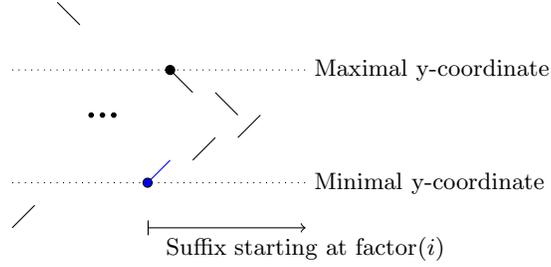

We now define the set
$\SET{\ppattern}{\ptext}{i}{j}$ 
as the set of every subsequence $s$ of $\ptext[j:]$ that starts at $\ptext[j]$ and
that is a matching of $\factor(i)$ $\factor(i-1)$ $\ldots$ $\factor(1)$ and .

\begin{lemma}
\label{lemma:ts}
Let $\ppattern$ be a permutation, 
$\factor(i)$ be an ascent (respectively descent) factor,
$s$ be a subsequence of $\pi$ such that $s \in \SET{\ppattern}{\ptext}{i}{j}$ and
that minimizes (resp. maximizes) the match of the left most element of $\factor(i-1)$.
For all subsequences 
$s' \in \SET{\ppattern}{\ptext}{i}{j}$ and for all subsequences $t$ of $\ptext$, such that $t=t's'$, 
if $t$ is a matching of $\ppattern[\LMEi(\factor(i+1)):]$ such that the left most element of $\factor(i)$ is matched to $\ptext[j]$ then the subsequence $t's$ is a matching of $\ppattern[\LMEi(\factor(i+1)):]$ such that the left most element of $\factor(i)$ is matched to $\ptext[j]$.
\end{lemma}

This lemma states that given any matching of $\factor(i+1)$ $\factor(i)$ $\ldots$ $\factor(1)$,
where $\factor(i)$ is an ascent (resp. descent) factor,
we can replace the part of the match that match $\factor(i)$ $\ldots$ $\factor(1)$ by any match
that minimise (resp. maximise) the left most element of $\factor(i-1)$. Indeed the left most element of $\factor(i-1)$ is the maximal (resp. minimal) element of  $\factor(i)$ $\ldots$ $\factor(1)$ (see figure \ref{fig:minimise the max}). 

\begin{proof}[of Lemma~\ref{lemma:ts}]
	By definition $s$ is a matching of $\ppattern[\LMEi(\factor(i)):]$. To prove that $ts$ is an matching of $\ppattern[\LMEi(\factor(i+1)):]$ we need to prove that every element of t is larger than every element of $s$. If $ts'$ is a matching of $\ppattern[\LMEi(\factor(i+1)):]$ then every element of $t$ is larger than every element of $s'$. Moreover the maximal element of $s$ is smaller than  the maximal element of $s'$ so every element of $s$ is smaller than every element of $s'$ thus every element of $s$ is smaller than every element of $t$. We use a similar argument if $\factor(i)$ is a descent factor.
	\qed
\end{proof}

\begin{corollary}
\label{corollary:we can chose a matching}	
Let $\ppattern$ be a permutation, 
$\factor(i)$ be an ascent (respectively descent) factor
and 
$s$ be a subsequence of $\pi$ such that $s \in \SET{\ppattern}{\ptext}{i}{j}$ and
that minimizes (resp. maximizes) the match of the left most element of $\factor(i-1)$.
These following statements are equivalent :
\begin{itemize}
	\item There exists a 
	matching of $\ppattern$ in $\ptext$ with the left most element of $\factor(i)$ is matched to $\ptext[j]$.
	\item There exists a matching $t$ of $\sigma[:\LMEi(\factor(i))-1]$ in $\ptext[:j-1]$  such that $ts$ is a matching of $\ppattern$ in $\ptext$ with the left most element of $\factor(i)$ matched to $\ptext[j]$.
\end{itemize}
\end{corollary}

This corollary takes a step further from the previous one, it states that if there is no matching using any match that maximise (resp. minimise) the left most element of $\factor(i-1)$ then there does not exist any matching at all. This is central to the algorithm because it allows to test only the matching that  maximise (resp. minimise) the left most element of $\factor(i-1)$ (see figure \ref{fig:minimise the max}).

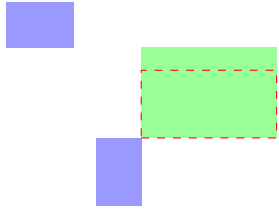
\begin{figure}[t] 
	\centering	
	\begin{tikzpicture}[scale=.3]
	
	\fill[blue!40!white] (0,9) rectangle (3,7);
	\fill[blue!40!white] (4,0) rectangle (6,3);
	\fill[green!40!white] (6,3) rectangle (12,7);

	\draw[red,dashed] (6,3) rectangle (12,6);
	
	\end{tikzpicture}
	
	\caption{In a matching, we can replace the green area by the red dashed area.}
	\label{fig:minimise the max}
\end{figure}


\begin{proposition}
	\label{Proposition:sigma avoids 213 and 231}
	Let $\sigma \in \AV_k(213,231)$ and $\pi \in S_n$.
	\\ One can decide in $O(max(kn^2,n^2\log(\log(n)))$ time
	and $O(kn^2)$ space if $\pi$ matches $\sigma$.
\end{proposition}

\begin{proof}

We first introduce a set of values needed to our proof.
Let $LIS_{\ptext}(i,j,bound)$ (resp. $LDS_{\ptext}(i,j,bound)$) be the longest increasing 
(resp. decreasing) sequence in $\ptext[i:j]$ starting at $i$,
with every element of this sequence
smaller (resp. bigger) than $bound$.
$LIS_{\ptext}$ and $LDS_{\ptext}$ can be computed in 
$O(n^2\log(\log(n)))$ time (see \cite{Bespamyatnikh00enumeratinglongest}).
As stated before, those values allow us to find a matching of a factor.

Now consider the following set of values (see figure \ref{fig: algo 1 fig}) :
$$
\LM{\ppattern}{\ptext}{i}{j} =
\begin{cases}

	\text{The match of $\LMEi(\factor(i-1))$ } 
		& \text{If $\factor(i)$ is} \\
	\text{ in a matching of $\sigma[\LMEi(\factor(i)):]$ in $\pi[j:]$} 
		& \text{an ascent factor} \\	
	\text{ and that minimizes the match } & \\
	\text{ $\LMEi(\factor(i-1))$} & \\
	\text{ and starts with $\ptext[j]$} & \\
	\text{Or $\infty$ if no such matching exists} & \text{} \\ 
		
	&\\

	\text{The match of $\LMEi(\factor(i-1))$ } 
	& \text{If $\factor(i)$ is} \\
	\text{ in a matching of $\sigma[\LMEi(\factor(i)):]$ in $\pi[j:]$} 
	& \text{a descent factor} \\	
	\text{ and that maximizes the match } & \\
	\text{ $\LMEi(\factor(i-1))$} & \\
	\text{ and starts with $\ptext[j]$} & \\
	\text{Or $\infty$ if no such matching exists} & \text{} \\

\end{cases}
$$

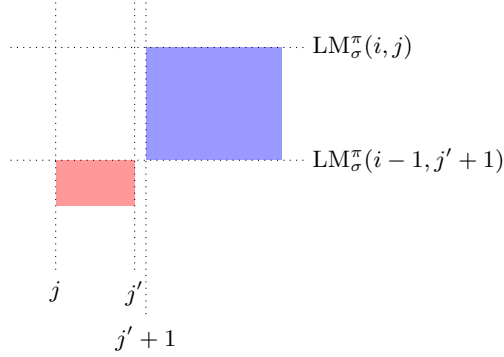
\begin{figure}[t] 
	\centering	
	\begin{tikzpicture}[scale=.3]
	
	\fill[red!40!white] (2,3) rectangle (5.5,5);
	\fill[blue!40!white] (6,10) rectangle (12,5);

	\draw[dotted,-] (0,10) -- (13,10) node[right] {$\LM{\ppattern}{\ptext}{i}{j}$};
	
	\draw[dotted,-] (0,5) -- (13,5) node[right] {$\LM{\ppattern}{\ptext}{i-1}{j'+1}$};

	\draw[dotted,-] (2,12) -- (2,0) node[below] {$j$};
	\draw[dotted,-] (5.5,12) -- (5.5,0) node[below] {$j'$};
	\draw[dotted,-] (6,12) -- (6,-2) node[below] {$j'+1$};
	
	
	\end{tikzpicture}

	\caption{
		When looking for a matching of $\sigma[\LMEi(\factor(i)):]$ (red and blue areas) starting at $\pi[j:]$, (1) the red area is a matching of $\factor(i)$ ie. if and only if
		its contains an increasing subsequence of size equal or bigger than $|\factor(i)|$.
		(2) the red and blue areas have to be "compatible" (their x-coordinates and y-coordinates have to be disjoint and the red area has to be before the blue area in the x-coordinate and y-coordinate).
		$\AF{\ppattern}{\ptext}{i}{j}$ is the set of every top y-coordinate of every blue area which is compatible. $\LM{\ppattern}{\ptext}{i}{j}$ is the minimal value of $\AF{\ppattern}{\ptext}{i}{j}$ if he set is not void.}	
	\label{fig: algo 1 fig}	
\end{figure}

Clearly there exists a matching of $\ppattern$ in $\ptext$ if and only if
there exists a $1 \leq i \leq n$ such that $LM(n_{factors},i)$ $\neq$ $0$ and $LM(n_{factors},i)\neq \infty$
with $n_{factors}$ the number of factor in $\ppattern$.
We show how to compute recursively those values.

\noindent\textbf{BASE :}
\begin{align*}
\LM{\ppattern}{\ptext}{1}{j} 
&=
\begin{cases}
		\min_{j<j'}\{\infty\} \cup \{\ptext[j'] |\text{ $j'$ such that \indent } 
			& \text{If $\factor(i)$ is}  \\
		\text{ $|\factor(1)| \leq LIS(j,j',\ptext[j']+1)  $}\}
			&\text{ an ascent factor} \\
		&\\
		\max_{j<j'}\{0\} \cup \{\ptext[j'] |\text{ $j'$ such that \indent } 
			& \text{If $\factor(i)$ is}  \\
		\text{ $|\factor(1)| \leq LDS(j,j',\ptext[j']-1) $}\}
			&\text{ a descent factor}\\
\end{cases}
\\
\intertext{
In the base case,
one is looking for a matching of the first factor.
}
\\
\intertext{\textbf{STEP :}}
\LM{\ppattern}{\ptext}{i}{j} 
&=
\begin{cases}
	\min \{\infty\} \cup  \AF{\ppattern}{\ptext}{i}{j} &
	\text{If $\factor(i)$ is an ascent factor}\\
	\max \{0\} \cup  \DF{\ppattern}{\ptext}{i}{j} &
	\text{If $\factor(i)$ is a descent factor}\\
\end{cases}
\end{align*}
where $\AF{\ppattern}{\ptext}{i}{j}$ and $\DF{\ppattern}{\ptext}{i}{j}$ are the sets 
of elements matching the left most element of $\factor(i-1)$ in a match of $\sigma[\LMEi(\factor(i)):]$ starting at $\ptext[j:]$.
Suppose that $\factor(i)$ is an ascent (resp descent) factor, index $j'$ and matching $t$ exists if and only if
$\ptext[j:j']$ contains a matching of $\factor(i)$ "compatibles" with a matching in $\ptext[j'+1:]$ of $\sigma[\LMEi(\factor(i-1)):]$. It is enough to assure that every element of the matching of $\factor(i)$ are smaller (resp. bigger) than the element of the matching of $\sigma[\LMEi(\factor(i-1)):]$. 

Thus we can compute $\AF{\ppattern}{\ptext}{i}{j}$ and $\DF{\ppattern}{\ptext}{i}{j}$ as follows:
\begin{align*}
\AF{\ppattern}{\ptext}{i}{j}
&=
\text{$\{\ptext[j'+1] \;|\; j<j'<n$ and $\LM{\ppattern}{\ptext}{i-1}{j'+1} \neq 0$ and} 
\\
&\qquad 
\text{$|\factor(i)| \leq LIS_{\ptext}(j,j',\LM{\ppattern}{\ptext}{i-1}{j'+1})\}$} 
\\
\DF{\ppattern}{\ptext}{i}{j}
&=
\text{$\{\ptext[j'+1] \;|\; j<j'<n$ and $\LM{\ppattern}{\ptext}{i-1}{j'+1} \neq \infty$ and}
\\
&\qquad
\text{$|\factor(i)| \leq LDS_{\ptext}(j,j',\LM{\ppattern}{\ptext}{i-1}{j'+1})\}$}
\end{align*}

The number of factor is bound by $k$.
Every instance of $LIS_{\ptext}$ and $LDS_{\ptext}$ can be computed in $O(n^2\log(\log(n))$.
There are $n$ base cases that can be computed $O(n)$ time, thus computing every base cases
takes $O(n^2)$ time.
There are $kn$ different instance of $\AFa$ and each one of them take $O(n)$ time to compute, 
thus computing every instance of $\AFa$ takes $O(kn^2)$ time.
There are $kn$ different instance of $\LMa$ and each one of them take $O(n)$, because the size of an $\AFa$
is bounded by $n$, thus computing every $\LMa$ takes $O(kn^2)$ time.
Thus computing all the values takes $O(max(kn^2,n^2\log(\log(n)))$.
Every value takes $O(1)$ space, thus the whole problem takes $O(kn^2)$ space.
\qed
\end{proof}


\section{$(213,231)$-avoiding bivincular patterns}
	\label{section:bivincular}

This section is devoted to the pattern matching problem with $(213,231)$-avoiding bivincular pattern.
Recall that a bivincular pattern generalises a permutation pattern by
being able to force elements to be consecutive in value or in position.
Hence,  bivincular pattern is stronger in constraint than permutation pattern, intuitively we can not use the previous algorithm.
As in a $(213,231)$-avoiding permutation, we can describe structural property of a $(213,231)$-avoiding bivincular pattern.

\begin{lemma}
\label{lemma:ascentDescentAscent}
Given $\ppattern$ a $(213,231)$-avoiding bivincular pattern,
If $\overline{\sigma[i]\sigma[j]}$ (this implies that $\sigma[j]=\sigma[i]+1$) and 
if $\sigma[i]$ is an ascent (resp. decent) element and $\sigma[i]+1$ is an ascent 
(resp. decent) element then :  
\begin{itemize}
	\item $i<j$ (resp. $j>i$)
	\item For every $l$, $i<l<j$ (resp. $j>l>i$), $\pi[l]$ is a descent (resp. ascent) element.    
\end{itemize}
\end{lemma}

This lemma states that if two ascent (resp. decsent) elements
need to be matched to consecutive elements in value then every element between those two elements (if any) is a descent (resp. ascent) element. 

\begin{proof}[of Lemma~\ref{lemma:ascentDescentAscent}]
	Suppose that there exists $l$, $i<l<j$,
	such that $\ppattern[l]$ is ascent. Ascent elements are increasing  so $\sigma[i]<\sigma[l]<\sigma[j]$ which is in contradiction with $\sigma[j]=\sigma[i]+1$.
	We use a similar argument if $\sigma[i]$ is a descent element
	\qed
\end{proof}

\begin{proposition}
\label{Proposition:bivincular pattern}
Let $\sigma$ be a $(213,231)$-avoiding bivincular of length $k$
and $\pi \in S_n$.
One can decide in $O(kn^4)$ time
and $O(kn^3)$ space if $\pi$ matches $\sigma$.
\end{proposition}

\begin{proof}

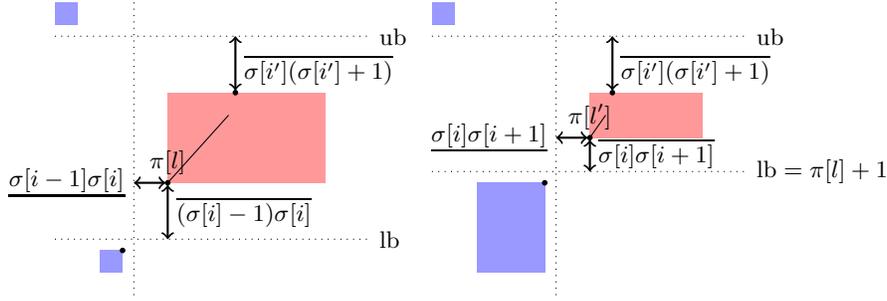
\begin{figure}[t] 
	\centering	
	\begin{tikzpicture}[scale=.3]
	\fill[blue!40!white] (0,12) rectangle (1,11);
	\draw[dotted,-] (0,10.5) -- (14,10.5) node[right] {$\ub$};

	\fill[blue!40!white] (2,0) rectangle (3,1);
	\draw[dotted,-] (0,1.5) -- (14,1.5) node[right] {$\lb$};
	
	\draw[fill=black] (3,1) circle (1mm) ;
	
	\draw[dotted,-] (3.5,12) -- (3.5,-1);
	
	\fill[red!40!white] (5,8) rectangle (12,4);
	\draw (5,4) -- (7.7,7);
	
	\draw[fill=black] (5,4) circle (1mm) node[above] {$\pi[l]$} ;
	
	\draw[thick,<->] (5,3.9) -- (5,1.5) ;	
	\draw (5,2.7) node[right] {$\overline{(\sigma[i]-1)\sigma[i]}$};
	
	\draw[thick,<->] (4.9,4) -- (3.5,4)  node[left] {$\underline{\sigma[i-1]\sigma[i]}$};	
	
	\draw[fill=black] (8,8) circle (1mm) ;
	\draw[thick,<->] (8,8.1) -- (8,10.5) ;	
	\draw (8,9) node[right] {$\overline{\sigma[i'](\sigma[i']+1)}$};

	\end{tikzpicture}		
	\begin{tikzpicture}[scale=.3]
	\fill[blue!40!white] (0,12) rectangle (1,11);
	\draw[dotted,-] (0,10.5) -- (14,10.5) node[right] {$\ub$};
	
	\fill[blue!40!white] (2,0) rectangle (5,4);
	\draw[dotted,-] (0,4.5) -- (14,4.5) node[right] {$\lb=\pi[l]+1$};
	
	
	\draw[dotted,-] (5.5,12) -- (5.5,-1);
	
	\fill[red!40!white] (7,6) rectangle (12,8);
	\draw (7,6) -- (7.7,7);

	\draw[fill=black] (5,4) circle (1mm) ;
	
	\draw[fill=black] (7,6) circle (1mm) node[above] {$\pi[l']$};

	\draw[thick,<->] (7,5.9) -- (7,4.5) ;	
	\draw (7,5.3) node[right] {$\overline{\sigma[i]\sigma[i+1]}$};
	
	\draw[thick,<->] (5.5,6)   node[left] {$\underline{\sigma[i]\sigma[i+1]}$} -- (6.9,6) ;	
	
	\draw[fill=black] (8,8) circle (1mm) ;
	\draw[thick,<->] (8,8.1) -- (8,10.5) ;	
	\draw (8,9) node[right] {$\overline{\sigma[i'](\sigma[i']+1)}$};

	\end{tikzpicture}	
	
	\caption{When matching an ascent element, the value of $ub$ does not change,
		and the lb is equal to the last match plus one. Note that there is only ascent element before $\sigma[i']$}	
	\label{fig:ub dont change}
\end{figure}

We consider the following set of boolean :
Given $\ppattern$ a $(213,231)$-avoiding bivincular pattern, and a text $\ptext \in \Av{n}{231,213}$,
$$
\PM{\sigma}{\ptext}{\lb}{\ub}{i}{j}=
\begin{cases}
	true 	& \text{If $\ptext[j:]$ matches $\sigma[i:]$ }\\
			& \text{with every element of the matching is in $[\lb,\ub]$}\\
			& \text{and starting $\pi[j]$}\\
	false 	& otherwise\\
\end{cases}
$$

The argument $\lb$ (respectively $\ub$) stands for 
the match of the last ascent (resp. decsent) element matched
plus (resp. minus) one.
We now show how to compute recursively those boolean (see fig. \ref{fig:ub dont change}).

\noindent\textbf{BASE:} \\
$$
\PM{\sigma}{\ptext}{\lb}{\ub}{k}{j}=
\begin{cases}
	true 	& \text{if $\ptext[j] \in [\lb,\ub ]$}\\
			& \text{and if ${\ppattern[k]}_\lrcorner$ then $j=n$}\\
			& \text{and if ${\ppattern[k]}^\urcorner$ then $\ptext[j]=\ub=k$}\\
			& \text{and if  $^\ulcorner{\ppattern[k]}$ then $\ptext[j]=\lb=1$ } \\
			& \text{and if  $\overline{(\ppattern[k]-1)\ppattern[k] }$ then $\ptext[j]=\lb$ }  \\
			& \text{and if  $\overline{(\ppattern[k]\ppattern[k]-1)}$ then $\ptext[j]=\ub$}  \\

	false	& \text{otherwise} \\
\end{cases}
$$

The base case finds an matching for the rightmost element of the pattern. If the last element does not have any restriction on positions and on values, then $\PM{\sigma}{\ptext}{\lb}{\ub}{k}{j}$ is true if and only if $\ppattern[k]$ is matched $\ptext[j]$. Which is true if $\ptext[j] \in [\lb,\ub]$. If ${\ppattern[k]}_\lrcorner$ then $\ppattern[k]$ must be matched to the right most element of $\pi$ thus $j$ must be the $n$. If ${\ppattern[k]}^\urcorner$ then $\ppattern[k]$ must be matched to the largest element which is $k$. If $^\ulcorner{\ppattern[k]}$ then $\ppattern[k]$ must be matched to the smallest element which is $1$. 
If  $\overline{(\ppattern[k]-1)\ppattern[k] }$ then the matched element of $\ppattern[k]$ and $\ppattern[k]-1$ must be consecutive in value, by recursion the value of the matched element of $\ppattern[k]-1$ will be recorded in $\lb$ and by adding $1$ to it thus $\ppattern[k]$ must be matched to $\lb$. 
If  $\overline{(\ppattern[k]\ppattern[k]-1)}$ then the matched element of $\ppattern[k]$ and $\ppattern[k]-1$ must be consecutive in value, by recursion the value of the matched element of $\ppattern[k]-1$ will be recorded in $\ub$ and by removing $1$ to it thus $\ppattern[k]$ must be matched to $\ub$. 

\noindent\textbf{STEP:}

We need to consider 3 cases for the problem $\PM{\sigma}{\ptext}{\lb}{\ub}{i}{j}$ :
\begin{itemize}
	\item If $\ptext[j] \notin [\lb,\ub]$ then :
	$$
	\PM{\sigma}{\ptext}{\lb}{\ub}{i}{j} = false
	$$

	which is immediate from the definition. If $\ptext[j] \notin [\lb,\ub]$ then it can not be part of a matching of  $\ppattern[i:]$ in $\ptext[i:]$ with every matched element in $[\lb,\ub]$.
	
	\item If $\ptext[j] \in [\lb,\ub]$ and $\ppattern[i]$ is an ascent element then :
	$$
	\PM{\sigma}{\ptext}{\lb}{\ub}{i}{j}=
	\begin{cases}
		\bigcup_{j<l} \PM{\sigma}{\ptext}{\ptext[j]+1}{\ub}{i+1}{l}
			& \text{if $\ppattern[i]$ is not underlined } \\
			& \text{ and $\ppattern[i]$ is not overlined} \\
		\bigcup_{j<l} \PM{\sigma}{\ptext}{\ptext[j]+1}{\ub}{i+1}{l}
			& \text{if $\ppattern[i]$ is not underlined } \\
			& \text{ and $\overline{(\ppattern[i]-1)\ppattern[i]}$ or $^\ulcorner{\ppattern[i]}$}\\
			& \text{ and $\ptext[j]=\lb$} \\
		\PM{\sigma}{\ptext}{\ptext[j]+1}{\ub}{i+1}{j+1}
			& \text{if $\underline{\ppattern[i]\ppattern[i+1]}$ } \\
			& \text{ and $\overline{(\ppattern[i]-1)\ppattern[i]}$ or $^\ulcorner{\ppattern[i]}$}\\
			& \text{ and $\ptext[j]=\lb$} \\
		\PM{\sigma}{\ptext}{\ptext[j]+1}{\ub}{i+1}{j+1}
			& \text{if $\underline{\ppattern[i]\ppattern[i+1]}$ } \\
			& \text{ and $\ppattern[i]$ is not overlined} \\
		false & \text{otherwise}

	\end{cases}
	$$	
	Remark that $\ptext[j]$ can be matched to $\ppattern[i]$ because $\ptext[j] \in [\lb,\ub]$.
	Thus if $\ptext[j+1:]$ matches $\ppattern[i+1:]$  with every element of the matching in $[\ppattern[i]+1,\ub]$ then 
 	$\ptext[j:]$ matches $\ppattern[i:]$. The last condition correspond 
 	to know if there exists $l$, $j<l$ such that
	$\PM{\sigma}{\ptext}{\ptext[j]+1}{\ub}{i+1}{l}$ is true.
	The first case correspond to a matching without restriction on position and on value. The second case asks for the match of $\ppattern[i]-1$ and $\ppattern[i]$ to be consecutive in value, but the match of $\ppattern[i]-1$
	is $lb-1$ thus we want $\ptext[j]=\lb$. The fourth case asks for the match of $\ppattern[i]$ and $\ppattern[i+1]$ to be consecutive in index, thus the match of $\ppattern[i+1]$ must be $j+1$. The third case is an union of the second and fourth case. Note that if $\pi[j]$ is an ascent element we can not have the condition that the match of $\ppattern[i]$ and $\ppattern[i]+1$ have to consecutive in value.
	
	\item If $\ptext[j] \in [\lb,\ub]$ and $\ppattern[i]$ is a descent element then :

	$$
	\PM{\sigma}{\ptext}{\lb}{\ub}{i}{j}=
	\begin{cases}
			\bigcup_{j<l} \PM{\sigma}{\ptext}{\lb}{\ptext[j]-1}{i+1}{l}
				& \text{if $\ppattern[i]$ is not underlined } \\
				& \text{ and $\ppattern[i]$ is not overlined} \\
			\bigcup_{j<l} \PM{\sigma}{\ptext}{\lb}{\ptext[j]-1}{i+1}{l}
				& \text{if $\ppattern[i]$ is not underlined } \\
				& \text{ and $\overline{\ppattern[i](\ppattern[i]+1)}$ or ${\ppattern[i]}^\urcorner$}\\
				& \text{ and $\ptext[j]=\ub$} \\
			\PM{\sigma}{\ptext}{\lb}{\ptext[j]-1}{i+1}{j+1}
				& \text{if $\underline{\ppattern[i]\ppattern[i+1]}$ } \\
				& \text{ and $\overline{\ppattern[i](\ppattern[i]+1)}$ or ${\ppattern[i]}^\urcorner$}\\
				& \text{ and $\ptext[j]=\ub$} \\
			\PM{\sigma}{\ptext}{\lb}{\ptext[j]-1}{i+1}{j+1}
				& \text{if $\underline{\ppattern[i]\ppattern[i+1]}$ } \\
				& \text{ and $\ppattern[i]$ is not overlined} \\
			false & \text{otherwise}
	\end{cases}
	$$	
	The same remark as the last case holds.

\end{itemize}

Clearly if $\bigcup_{0<j} \PM{\sigma}{\ptext}{1}{n}{1}{j}$ is true then $\pi$ matches $\sigma$. We now discuss the position and value constraints.
\paragraph{Position Constraint.} There are 3 types of position constraints that can be added by underlined elements.
\begin{itemize}
	\item If $_\llcorner{\sigma[1]}$ then the leftmost element of $\sigma$  must be matched to the leftmost element of $\pi$ ($\ppattern[1]$ is matched to $\ptext[1]$ on a matching of $\ppattern$ in $\ptext$). This constraint is satisfied by requiring that the matching starts at the left most element of $\ptext$ : if  $\PM{\sigma}{\ptext}{1}{n}{1}{1}$ is true.
	
	\item If ${\ppattern[k]}_\lrcorner$ then the rightmost element $\sigma$ must be matched the rightmost element of $\pi$ ($\ppattern[k]$ is matched to $\ptext[n]$ on a matching of $\ppattern$ in $\ptext$). This constraint is checked in the base case.

	\item If $\underline{\ppattern[i]\ppattern[i+1]}$ then the index of the matched elements of $\ppattern[i]$ and $\ppattern[i+1]$ must be consecutive. In other word, if $\ppattern[i]$ is matched to $\ptext[j]$ then $\ppattern[i+1]$ must be matched to $\ptext[j+1]$. We assure this restriction by recursion by requiring that the matching of $\ppattern[i+1:]$ to start at index $j+1$ (see figure $\ref{fig:ub dont change}$).
\end{itemize}

\paragraph{Value Constraint.} There are 3 types of position constraints that can be added by overlined elements.
\begin{itemize}
	\item If $^\ulcorner{\sigma[i]}$ (and thus $\sigma[i]=1$) then the minimal value of $\ppattern$ must be matched to the minimal value of $\ptext$.
	\begin{itemize}

		\item If $\sigma[i]$ is an ascent element, then remark that 
		every problem 
		$\PM{\sigma}{\ptext}{\lb}{*}{i}{*}$ is true if $\sigma[i]$ is matched to element with value $\lb$ (by recursion) thus it is enough to require that $\lb=1$.
		Now remark that $\sigma[i]$ is the leftmost ascent element, indeed if not, then there exists an ascent element $\sigma[i']$, $i'<i$ and by definition $\sigma[i']<\sigma[i]$ which is not possible as $\sigma[i]$ must be the minimal element. 
		As a consequence $\sigma[1]$, $\ldots$, $\sigma[i-1]$ are descent elements.
		Moreover the recursive call from a descent element does not modified the lower bound
		thus for every $\PM{\sigma}{\ptext}{\lb}{*}{i}{*}$, $\lb=1$ (see figure $\ref{fig:ub dont change}$). 			

		\item If $\sigma[i]$ is a descent element then $i=k$ ($\sigma[i]$ is the right most element). Thus every $\PM{\sigma}{\ptext}{*}{*}{i}{*}$ is a base case and is true if $\sigma[i]$ is matched to $1$.

	\end{itemize}

	\item If ${\ppattern[i]}^\urcorner$ (and thus $\sigma[i]=k$) then the maximal value of $\ppattern$ must be matched to the maximal value of $\ptext$.
	\begin{itemize}

		\item If $\sigma[i]$ is an descent element, then remark that 
		every recursive call \\ $\PM{\sigma}{\ptext}{*}{\ub}{i}{*}$ is true if $\sigma[i]$ is matched to element with value $\ub$ (by recursion) thus it is enough to require that $\ub=n_{\ptext}$.
		Now remark that $\sigma[i]$ is the leftmost descent element, indeed if not, then there exists an descent element $\sigma[i']$, $i'<i$ and by definition $\sigma[i']>\sigma[i]$ which is not possible as $\sigma[i]$ must be the maximal element.
		As a consequence $\sigma[1]$, $\ldots$, $\sigma[i-1]$ are ascent elements.
		Moreover the recursive call from a ascent element does not modified the upper bound
		thus for every $\PM{\sigma}{\ptext}{*}{\ub}{i}{*}$, $\ub=n$ (see figure $\ref{fig:ub dont change}$.

		\item If $\sigma[i]$ is an ascent element then $\sigma[i]$ then $i=k$ ($\sigma[i]$ is the right most element). Thus every $\PM{\sigma}{\ptext}{*}{*}{i}{*}$ is a base case and is true if $\sigma[i]$ is matched to $n_{\ptext}$.

	\end{itemize}

	\item  If $\overline{\ppattern[i]\ppattern[i']}$,  (which implies that $\ppattern[i']=\ppattern[i]+1$) then if $\ppattern[i]$ is matched to $\ptext[j]$ then $\ppattern[i']$ must be matched to $\ptext[j]+1$.

		\begin{itemize}

			\item The case $\ppattern[i]$ is a descent element, $\ppattern[i']$ is an ascent element and $i<i'$ (remark that this case is equivalent to the case $\ppattern[i]$ is an ascent element, $\ppattern[i']$ is a descent element and $i'<i$) is not possible. 
			Indeed $\ppattern[i]$ is the maximal element of $\ppattern[i:]$ thus $\ppattern[i] > \ppattern[i']$ which is in contradiction with 
			$\ppattern[i']=\ppattern[i]+1$. 
			
			\item If $\ppattern[i]$ is an ascent element, $\ppattern[i']$ is a descent element and $i<i'$ (remark that this case is symmetric to the case $\ppattern[i]$ is a descent element, $\ppattern[i']$ is an ascent element and $i'<i$), then
			remark that every recursive call
			$\PM{\sigma}{\ptext}{\lb}{*}{i'}{*}$ is true if $\sigma[i']$ is matched to the element with $\lb$ (by recursion) thus it is enough to require that $\lb=\ptext[j]+1$.
			Now remark that
			$\ppattern[i]$ is the right most ascent element and $\ppattern[i']$ is the right most element (or $\ppattern[i'] \neq \ppattern[i]+1$). 
			As a consequence $\sigma[i+1]$, $\sigma[i+2]$, $\ldots$, $\sigma[i'-1]$ are descent elements. Moreover the recursive call from a descent element does not modified the lower bound and $\PM{\sigma}{\ptext}{\lb}{*}{i}{*}$ will put the lower bound to $\ptext[j]+1$
			thus for every $\PM{\sigma}{\ptext}{\lb}{*}{i'}{*}$, $\lb=\ptext[j]+1$ (see figure $\ref{fig:ub dont change}$. 	
			
			\item If $\ppattern[i]$ is an ascent element and $\ppattern[i']$ is an ascent element then
			first remark that every recursive call
			$\PM{\sigma}{\ptext}{*}{\ub}{i'}{*}$ is true if $\sigma[i']$ is matched to element with value $\lb$. 
			Now remark that

			$i<i'$ and there is no  ascent element between $\ppattern[i]$ and $\ppattern[i']$ (lemma \ref{lemma:ascentDescentAscent}), 
			As a consequence $\sigma[i+1]$, $\sigma[i+2]$, $\ldots$, $\sigma[i'-1]$ are descent elements.  Moreover the recursive call from a descent element does not modified the lower bound and $\PM{\sigma}{\ptext}{\lb}{*}{i}{*}$ will put the lower bound to $\ptext[j]+1$ thus for every $\PM{\sigma}{\ptext}{\lb}{*}{i'}{*}$, $\lb=\ptext[j]+1$ (see figure $\ref{fig:ub dont change}$.
		\end{itemize}
\end{itemize}

There are $n^3$ base cases that can be computed in constant time.
There are $kn^3$ different cases. Each case takes up to $O(n)$ time to compute.
Thus computing all the cases take $O(kn^4)$ time.
Each case take $O(1)$ space, thus we need $O(kn^3)$ space.
\qed
\end{proof}


\section{Computing the longest $(213,231)$-avoiding pattern}
\label{section:LCS}

This section is focused on a problem related to the pattern matching problem, finding the longest $(213,231)$-avoiding subsequences in permutations:
Given a set of permutations, find a longest $(213,231)$-avoiding that can be matched
by each input permutation.
This problem is know to be NP-Hard for an arbitrary number of permutations and we do not hope it to be solvable in polynomial time even with the constraint of the subsequence must avoid $(213,231)$. Thus we focus on the cases where only one or two permutations are given in input.
We start with the easiest case where we are given just one input permutation.
We need the set of descent elements and the set of ascent elements.
$A(\pi) = \{i | \text{$\pi[i]$ is an ascent element} \} \cup \{n\}$ and
$D(\pi) = \{i | \text{$\pi[i]$ is a descent element} \cup \{n\}$.\\

\begin{proposition}
\label{proposition:longestIncreasingSubsequence}
If $s$ is a longest $(213,231)$-avoiding subsequence with last element at index 
$f$ in $\pi$ then
$A(\pi)$ is a longest increasing subsequence with last element at index $f$ and
$D(\pi)$ is a longest decreasing subsequence with last element at index $f$.
\end{proposition}

\begin{proof}[of Proposition~\ref{proposition:longestIncreasingSubsequence}]
Let $s$ be a longest subsequence avoiding (213,231) with last element at index $f$ in $\pi$,
suppose that $P(s)$ is not a longest increasing subsequence with last element at index $f$. Let $s_m$ be a longest increasing subsequence with last element $f$.
Thus $|s_m|>|A(s)|$, clearly the sequence $s_m \cup D(s)$
is $(213,231)$-avoiding and is longer than $s$, this is a contradiction.
The same idea can be used to show that $D(\pi)$ is the longest decreasing subsequence.
\qed
\end{proof}

\begin{proposition}
\label{proposition:longest 2}
Let $\pi$ be a permutation. One can compute
the longest $(213,231)$-avoiding subsequence that can be matched in $\pi$
in $O(n\log(\log(n)))$ time and in $O(n)$ space.
\end{proposition}

\begin{proof}[of Proposition~\ref{proposition:longest 2}]
The proposition \ref{proposition:longestIncreasingSubsequence} lead to an algorithm
where one has to compute longest increasing and decreasing subsequence ending at every index. Then finding the maximal sum of longest increasing and decreasing subsequence ending at the same index.
Computing the longest increasing subsequence and the longest decreasing subsequence can be done in 
$O(n\log(\log(n)))$ time and $O(n)$ space 
(see \cite{Bespamyatnikh00enumeratinglongest}), 
then finding the maximal can be done in linear time.
\qed
\end{proof}

We  now consider the case where the input is composed of two permutations.

\begin{proposition}
Given two permutations $\pi_1$ and $\pi_2$,
one can compute 
the longest common $(213,231)$-avoiding subsequence
in $O(|\pi_1|^3|\pi_2|^3)$ time and space.
\end{proposition}

\begin{proof}
Consider the following problem
that computes the longest stripe common to $\pi_1$ and $\pi_2$:
Given two permutations $\pi_1$ and $\pi_2$, we define
$
\LCS{\pi_1}{\lb_1}{\ub_1}{\pi_2}{\lb_2}{\ub_2}{i_1}{i_2}
$
\begin{center}
= $\max$ \{$|s|$ $|$ $s$ can be matched $\pi_1[i_1:]$ with every element of the match in $[\lb_1,\ub_1]$ and $s$ can be matched $\pi_2[i_2:]$ with every element of the match in $[\lb_2,\ub_2]$ \}
\end{center}

We show how to solve this problem by dynamic programming.

\noindent\textbf{BASE:}
$$
\LCS{\pi_1}{\lb_1}{\ub_1}{\pi_2}{\lb_2}{\ub_2}{|\pi_1|}{|\pi_2|} =
\begin{cases}
	1 & \text{if $\lb_1 \leq \pi_1[j] \leq \ub_1$
	}\\
	& \text{ and $\lb_2 \leq \pi_2[j] \leq \ub_2$}\\
	0 & otherwise\\
\end{cases}
$$

\noindent\textbf{STEP:}
$$
\LCS{\pi_1}{\lb_1}{\ub_1}{\pi_2}{\lb_2}{\ub_2}{i_1}{i_2}=max
\begin{cases}
	\LCS{\pi_1}{\lb_1}{\ub_1}{\pi_2}{\lb_2}{\ub_2}{i_1}{i_2+1} \\
	\\
	\LCS{\pi_1}{\lb_1}{\ub_1}{\pi_2}{\lb_2}{\ub_2}{i_1+1}{i_2} \\
	\\
	\match{\pi_1}{\lb_1}{\ub_1}{\pi_2}{\lb_2}{\ub_2}{i_1}{i_2}
\end{cases}
$$

with \\
$
\match{\pi_1}{\lb_1}{\ub_1}{\pi_2}{\lb_2}{\ub_2}{i_1}{i_2}=
\begin{cases}
1+\LCS{\pi_1}{\pi_1[i_1]+1}{\ub_1}{\pi_2}{\pi_2[i_2]+1}{\ub_2}{i_1}{i_2+1}
	& \text{$\pi_1[i_1]<\lb_1$ } \\
	& \text{and $\pi_2[i_2]<\lb_2$} \\

&\\

1+\LCS{\pi_1}{\lb_1}{\pi_1[i_1]-1}{\pi_2}{\lb_2}{\pi_2[i_2]-1}{i_1+1}{i_2}
	& \text{$\pi_1[i_1]>\ub_1$ } \\
	&\text{and $\pi_2[i_2]>\ub_2$}\\

&\\

0 	& \text{otherwise}\\
\end{cases}
$

For every pair $i,j$ we either ignore the element of $\pi_1$,
or we ignore the element of $\pi_2$,
or we match as the same step (if possible).
These relations lead to a $O(|\pi_1|^3|\pi_2|^3)$ time and $O(|\pi_1|^3|\pi_2|^3)$ space algorithm.
Indeed there is $|\pi_1|^3|\pi_2|^3$ cases possible for the problem
and each case is solved in constant time.
$\qed$
\end{proof}

\bibliography{bibli}{}

\begin{thebibliography}{10}

\bibitem{Ahal:Rabinovich:2008}
S.~Ahal and Y.~Rabinovich.
\newblock On complexity of the subpattern problem.
\newblock {\em SIAM Journal on Discrete Mathematics}, 22(2):629--649, 2008.

\bibitem{Albert:Aldred:Atkinson:Holton:ISAAC:2001}
M.H. Albert, R.E.L. Aldred, M.D. Atkinson, and D.A. Holton.
\newblock Algorithms for pattern involvement in permutations.
\newblock In {\em Proc. International Symposium on Algorithms and Computation
  (ISAAC)}, volume 2223 of {\em Lecture Notes in Computer Science}, pages
  355--366, 2001.

\bibitem{Bespamyatnikh00enumeratinglongest}
Sergei Bespamyatnikh and Michael Segal.
\newblock Enumerating longest increasing subsequences and patience sorting,
  2000.

\bibitem{Bona:ElJC:2012}
M.~B{\'o}na.
\newblock Surprising symmetries in objects counted by catalan numbers.
\newblock {\em The Electronic Journal of Combinatorics}, 19(1):P62, 2012.

\bibitem{Bose:Buss:Lubiw:1998}
P.~Bose, J.F.Buss, and A.~Lubiw.
\newblock Pattern matching for permutations.
\newblock {\em Information Processing Letters}, 65(5):277--283, 1998.

\bibitem{Bouvel:Rossin:Vialette:CPM:2007}
M.~Bouvel, D.~Rossin, and S.~Vialette.
\newblock Longest common separable pattern between permutations.
\newblock In B.~Ma and K.~Zhang, editors, {\em Proc. Symposium on Combinatorial
  Pattern Matching (CPM'07), London, Ontario, Canada}, volume 4580 of {\em
  Lecture Notes in Computer Science}, pages 316--327, 2007.

\bibitem{DBLP:journals/corr/abs-1204-5224}
Marie{-}Louise Bruner and Martin Lackner.
\newblock A fast algorithm for permutation pattern matching based on
  alternating runs.
\newblock {\em CoRR}, abs/1204.5224, 2012.

\bibitem{Crochemore:Porat:2010}
M.~Crochemore and E.~Porat.
\newblock Fast computation of a longest increasing subsequence and application.
\newblock {\em Information and Compututation}, 208(9):1054--1059, 2010.

\bibitem{Downey:Fellows:2013}
R.G. Downey and M.~Fellows.
\newblock {\em Fundamentals of Parameterized Complexity}.
\newblock Addison Wesley Longman Publishing Co., Inc., Redwood City, CA, USA,
  2013.

\bibitem{Griffiths:Smith:Warren:PMA:2011}
W.~Griffiths, R.~Smith, and D.Warren.
\newblock Almost avoiding pairs of permutations.
\newblock {\em Pure Mathematics and Applications}, 22(2):129--139, 2011.

\bibitem{Guillemot:Marx:SODA:2014}
S.~Guillemot and D.~Marx.
\newblock Finding small patterns in permutations in linear time.
\newblock In C.~Chekuri, editor, {\em Proceedings of the Twenty-Fifth Annual
  ACM-SIAM Symposium on Discrete Algorithms (SODA), Portland, Oregon, USA},
  pages 82--101. SIAM, 2014.

\bibitem{Guillemot:Vialette:ISAAC:2009}
S.~Guillemot and S.~Vialette.
\newblock Pattern matching for $321$-avoiding permutations.
\newblock In Y.~Dong, D.-Z. Du, and O.~Ibarra, editors, {\em Proc. $20$-th
  International Symposium on Algorithms and Computation (ISAAC), Hawaii, USA},
  volume 5878 of {\em LNCS}, page 1064–1073. Springer, 2009.

\bibitem{Ibarra:1997}
L.~Ibarra.
\newblock Finding pattern matchings for permutations.
\newblock {\em Information Processing Letters}, 61(6):293--295, 1997.

\bibitem{Kitaev:book:2011}
S.~Kitaev.
\newblock {\em Patterns in Permutations and Words}.
\newblock Springer-Verlag, 2013.

\bibitem{Knuth:1997:ACP:260999}
Donald~E. Knuth.
\newblock {\em The Art of Computer Programming, Volume 1 (3rd Ed.): Fundamental
  Algorithms}.
\newblock Addison Wesley Longman Publishing Co., Inc., Redwood City, CA, USA,
  1997.

\bibitem{Simion:Schmidt:EJC:1985}
R.~Simion and F.W.Schmidt.
\newblock Restricted permutations.
\newblock {\em European Journal of Combinatorics}, 6(4):383–406, 1985.

\end{thebibliography}
\bibliographystyle{plain}


\end{document}